\documentclass[12pt,aip,cha]{revtex4-1}

\usepackage{amsthm,amssymb,setstack}
\usepackage[T1]{fontenc}
\usepackage{times}
\usepackage{color,graphicx}
\usepackage{subfig,array}

\newcommand\ket[1]{\ensuremath{|#1\rangle}}
\newcommand\bra[1]{\ensuremath{\langle#1|}}

\newcommand\ip[2]{\ensuremath{\langle#1,#2\rangle}}

\newcommand{\range}{\mathop{\rm range}\nolimits}
\newcommand{\Span}{\mathop{\rm span}\nolimits}

\def\D{\mathcal{D}}

\def\P{\mathcal{P}}

\newtheorem{definition}{Definition}
\newtheorem{lemma}{Lemma}
\newtheorem{theorem}{Theorem}
\newtheorem{corollary}{Corollary}

\newtheorem{proposition}{Proposition}

\newtheorem{example}{Example}

\begin{document}

\title[Ground-State Spaces of Frustration-Free Hamiltonians]{Ground-State Spaces of Frustration-Free Hamiltonians}

\author{Jianxin Chen}
\affiliation{Department of Mathematics \& Statistics, University of
Guelph, Guelph, Ontario, Canada}
\affiliation{Institute for Quantum
Computing, University of Waterloo, Waterloo, Ontario, Canada}

\author{Zhengfeng Ji}
\affiliation{Institute for Quantum Computing, University of
Waterloo, Waterloo, Ontario, Canada}
\affiliation{State Key
Laboratory of Computer Science, Institute of Software, Chinese
Academy of Sciences, Beijing, China}

\author{David Kribs}
\affiliation{Department of Mathematics \& Statistics, University of
Guelph, Guelph, Ontario, Canada} \affiliation{Institute for Quantum
Computing, University of Waterloo, Waterloo, Ontario, Canada}

\author{Zhaohui Wei}
\email{cqtwz@nus.edu.sg} \affiliation{Centre for Quantum
Technologies, National University of Singapore, Singapore}

\author{Bei Zeng}
\affiliation{Department of Mathematics \& Statistics, University of
Guelph, Guelph, Ontario, Canada} \affiliation{Institute for Quantum
Computing, University of Waterloo, Waterloo, Ontario, Canada}

\begin{abstract}
We study the ground-state space properties for frustration-free Hamiltonians. We introduce a
concept of `reduced spaces' to characterize local structures of ground-state spaces.
For a many-body system, we characterize mathematical structures for the set $\Theta_k$ of all the $k$-particle reduced spaces, which with a binary operation called join forms a semilattice that can be interpreted as an abstract convex structure. The smallest nonzero elements in $\Theta_k$, called atoms, are analogs of extreme points.
We study the properties of atoms in $\Theta_k$ and discuss its relationship with ground states of $k$-local frustration-free Hamiltonians. For spin-$1/2$ systems, we show that all the atoms in $\Theta_2$ are unique ground states of some $2$-local frustration-free Hamiltonians. Moreover, we show that the elements in $\Theta_k$ may not be the join of atoms, indicating a richer structure for  $\Theta_k$ beyond the convex structure. Our study of $\Theta_k$ deepens the understanding of ground-state space properties for frustration-free Hamiltonians, from a new angle of reduced spaces.
\end{abstract}



\maketitle

\section{Introduction}
\label{sec:intro}

Quantum many-body physics is no doubt one of the most exciting areas
in modern physics. The interaction, correlation or entanglement
between particles result in intriguing physical systems such as
superconductors and topological insulators, which are materials with
dramatically different physical properties from the traditional
matters such as conductors and insulators~\cite{HK10}. Given a
quantum many-body system of $n$ particles, described by a
Hamiltonian $H$, the first basic feature to understand is its ground
state property, or, its ground-state space property when the ground
states are degenerate. In practical systems, the Hamiltonian $H$ can
be written as $H=\sum_j H_j$, where each term $H_j$ acts on at most
$k$-particles. This kind of Hamiltonian is called a $k$-local
Hamiltonian~\cite{Has10}. In many physical systems, one has $k=2$ as
the $H_j$'s involve at most two-particle interactions.

The ground state energy of a $k$-local Hamiltonian, given by
$E_0=\bra{\psi_0}H\ket{\psi_0}$, where $\ket{\psi_0}$ is a ground
state, can just be given by the $k$-particle reduced density
matrices ($k$-RDMs) of $\ket{\psi_0}$. For a many-body system, the
set of all $k$-RDMs, denoted by $\mathcal{D}_k$, is a closed convex
set. The knowledge of $\mathcal{D}_k$ can replace the many-body wave
functions by their $k$-RDMs, in calculating physical observables,
such as the ground state energy. Therefore, it is highly desired to
find a full characterization for the structure of $\mathcal{D}_k$,
in particular for $k=2$. Unfortunately, this is shown to be hard
even with the existence of a quantum computer~\cite{Liu06}. However,
a  better understanding of the geometry of $\mathcal{D}_k$ can
provide more practical information for calculations over
$\mathcal{D}_k$.

An important type of local Hamiltonians widely studied are called
`frustration-free' Hamiltonians, where for $H=\sum_j H_j$, the
ground-state space of $H$ is also in the ground-state space of each
$H_j$~\cite{PVC+07,BT08}. This `frustration-free' property, though
at first glance might seem very restricted, is actually complicated
enough to give rise to many interesting physical phenomena, such as
the fractional quantum Hall effect and topological phase of
matter~\cite{KL09}. Studies in quantum information science show that
in general, to find out whether a $k$-local Hamiltonian is
frustration-free is very hard, even with the existence of a quantum
computer~\cite{Bra06}. It turns out that the only frustration-free
system that is relatively easy to understand is the case $k=2$ for
qubit systems, where the ground state energy can be obtained with a
polynomial algorithm and the ground-state space structure can be
completely characterized by a tree tensor network
structure~\cite{Bra06,CCD10,JWZ10}.

Note that for any ground state $\ket{\psi_0}$ of a $k$-local
frustration-free (FF) Hamiltonian $H$, any other state with the same
ranges of the $k$-RDMs as those of $\ket{\psi_0}$ is also a ground
state of $H$. In other words, unlike the case for general $k$-local
Hamiltonians, where full information of the $k$-RDMs is needed to
determine the ground-state space properties, for the case of
$k$-local FF Hamiltonians, only the information of the ranges is
sufficient. Here the range of a density operator $\rho$ is the space
spanned by all the eigenvectors of $\rho$ corresponding to nonzero
eigenvalues. Given that the range of a density operator is a
subspace, we name the ranges of $k$-RDMs by `$k$-particle reduced
spaces' ($k$-RSs), and denote the set of all the $k$-RSs by
$\Theta_k$. It is then natural to ask what is the mathematical
structure that can characterize $\Theta_k$. Similar to the general
case with frustration, this mathematical structure will provide
useful information for studying the ground-state space properties
for $k$-local FF Hamiltonians.

Recall that the set of all the subspaces of the Hilbert space is a lattice, whose structure has been widely studied in the field of quantum logic~\cite{Kal83}. In this paper we show that $\Theta_k$ is a semilattice, and interestingly, this semilattice structure can be viewed as an abstract convex structure~\cite{Fri09}. This abstract convex structure of $\Theta_k$ can be then viewed as an analog of the convex set $\mathcal{D}_k$. The smallest nonzero elements in $\Theta_k$, called atoms, are analogs of the extreme points in $\mathcal{D}_k$. We study the properties of atoms in $\Theta_k$ and discuss the relationship with ground states of $k$-local FF Hamiltonians. For a spin-$1/2$ system, we show that all the atoms in $\Theta_2$ are unique ground states of $2$-local FF Hamiltonians.  Moreover, contrary to the points in $\mathcal{D}_k$ which can always be weighted sums of extreme points, we show that elements in $\Theta_k$ may not be the join of atoms, indicating a richer structure of $\Theta_k$. We believe our study of $\Theta_k$ deepens the understanding of ground-state space properties for frustration-free Hamiltonians.

We organize our paper as follows. In Sec.~\ref{sec:Dk} we recall
some background information regarding the convex structure of
$\mathcal{D}_k$ and the relationship between the geometry of
$\mathcal{D}_k$ and the ground-state spaces of $k$-local
Hamiltonians. In Sec.~\ref{sec:kRS}, we give the formal definition
of $\Theta_k$. In Sec.~\ref{sec:semigroup}, we introduce a binary
operation of the elements, under which $\Theta_k$ is closed and
forms a semigroup. In Sec.~\ref{sec:semilattice}, we then show that
$\Theta_k$ equipped with this binary operation is a semilattice. We
further study in detail this semilattice structure, with especial
focus on its smallest elements, called atoms, and its building
blocks called join irreducibles.  In Sec.~\ref{sec:convex}, we show
that this semilattice can be viewed as an abstract convex structure,
with atoms being analogs of extreme points.  In Sec.~\ref{sec:GSS}
we relate the structure of $\Theta_k$ to the ground-state spaces of
$k$-local FF Hamiltonians. Finally, a conclusion and discussion is
given in Sec.~\ref{sec:con}.

\section{The Convex Set of Reduced Density Matrices}
\label{sec:Dk}

In this section, we recall some background materials regarding the convex structure of $\mathcal{D}_k$ and the relationship between the geometry of $\mathcal{D}_k$ and the ground-state spaces of $k$-local Hamiltonians.

Consider a $n$-particle system. Let $\D$ be the set of density
matrices of $n$-particle states. For $\rho\in\D$, let
$\vec{R}_{k}(\rho) =(\gamma_j) = ( \gamma_1, \gamma_2, \ldots,
\gamma_m)$ be the vector whose elements $\gamma_{j}$s are the
$k$-RDMs of $\rho$ in a fixed order, where $m={n\choose k}$. This
then gives a map
\begin{equation}
\mathcal{\vec{R}}_{k}: \rho\mapsto (\gamma_j),
\end{equation}
where $\rho$ is then a pre-image of $(\gamma_j)$ under
$\mathcal{\vec{R}}_{k}$. Note that the map $\mathcal{\vec{R}}_{k}$
is not one-to-one. That is, there may exist some other $\rho'\in\D$
such that $\vec{R}_{k}(\rho')$=$\vec{R}_{k}(\rho)$. Consequently,
the set of pre-images of $\vec{R}_{k}(\rho)$ contain $\rho$. Denote
$\D_k = \{\vec{R}_k(\rho) \,|\, \rho\in\D\}$ the set of all the
$k$-RDMs. Then it is straightforward to see that $\D_k$ is a closed
convex set.

To relate the geometry of $\D_k$ to ground-state spaces of $k$-local Hamiltonians $H=\sum_j H_j$,
we need the concepts of dual cone and face.
\begin{definition}
For a convex set $C$ and a set $E$, the dual cone of $C$ with
respect to $E$ is defined to be
\begin{equation}
  \P(C) = \left\{ y \,|\, \forall\, x\in C,\ y\in E, \ip{x}{y} \ge 0
  \right\},
\end{equation}
where $\ip{x}{y}$ is some kind of inner product.
\end{definition}
Let the dual cone of the closed convex set $\D_k$ be $\P_k$ where,
for $x=(\gamma_j)$ and $y=(H_j)$, $\ip{x}{y}$ is defined as $\sum_j
Tr(H_j\gamma_j)$. Any point $\vec{H}$ of the form $(H_j)$ defines a
$k$-local Hamiltonian $H = \sum_j H_j$. Moreover, we have
$\ip{\vec{R}_k(\rho)}{\vec{H}} = Tr(\rho H)$. This allows us to
visualize $k$-local Hamiltonians as hyperplanes in the space
containing $\D_k$.

We now recall the concept of face.
\begin{definition}~\cite{Roc96}
For any convex set $C$, a subset $F$ is called a {\em face} on $C$ if
\begin{enumerate}
\item $F$ is a convex set.
\item For any line segment $L\subseteq C$,
if $L$ intersects $F$ at some point other than the two end points of
$L$, then $L\subseteq F$.
\end{enumerate}
\end{definition}

As a $k$-local Hamiltonian $H$ is a hyperplane in the space
containing $\D_k$, the face $F$ of $\D_k$ that this hyperplane
touches then corresponds to the ground-state space of $H$. In other
words, the sum of the ranges of the pre-images of all the points in
$F$ under the map $\mathcal{\vec{R}}_{k}$ gives the ground-state
space of $H$. These kind of faces are called exposed. However, for
any given face $F$ of $\D_k$, there might not exist a hyperplane
which touches only $F$ but no other points of $\mathcal{D}_k$. In
other words, in general there may exist non-exposed faces.

Note that the extreme points of $\D_k$ are zero-dimensional faces.
Given an extreme point $P$, if there exists a hyperplane (i.e. a
$k$-local Hamiltonian $H$), which touches only the point $P$ but no
other points of $\mathcal{D}_k$, then $P$ is called an exposed
point. In this case, the pre-images of $P$ under the map
$\mathcal{\vec{R}}_k$ give the ground-state space of $H$. If the
pre-image of $P$ is unique (i.e. a pure state $\ket{\psi}$), then
$\ket{\psi}$ is a unique ground state of $H$. Similarly, in general
there might exist non-exposed extreme points.

\section{The $k$-particle Reduced Spaces}
\label{sec:kRS}

As discussed in the introduction, to study the ground-state
properties for $k$-local FF Hamiltonians, we do not need the full
information of $k$-RDMs. Indeed, only the ranges of the $k$-RDMs are
needed. In this section, we will introduce the concept of
$k$-particle reduced spaces to characterize these ranges.

Let $\Theta$ be the set of subspaces of the $n$-particle Hilbert space.
Consider an $n$-particle subspace $\mathcal{S}\in \Theta$, which is non-empty,
namely, the dimension of $\mathcal{S}$ is at least one. Define
\begin{equation}
F_{\mathcal{S}}=\{\rho|\range(\rho)\subseteq \mathcal{S}\},
\end{equation}
which is a set of all the $n$-particle density operators, whose
ranges are subspaces of $\mathcal{S}$.

Recall that for any $n$-particle state $\rho$, the $k$-RDMs of $\rho$ is given by $\vec{R}_{k}(\rho)=(\gamma_1,\ldots,\gamma_m)$.
For each $\gamma_j$, let
\begin{equation}
\eta_j=\sum_{\rho\in F_{\mathcal{S}}}\range(\gamma_j),
\end{equation}
where the sum is the usual sum of vector spaces. That is, $\eta_j$
is the sum of the ranges of the $j$-th element $\gamma_j$ of the
$k$-RDMs of all possible $\rho$'s, whose ranges are subspaces of
$\mathcal{S}$.
\begin{definition}
The $k$-particle reduced spaces ($k$-RSs) of $\mathcal{S}$, denoted by $\vec{L}_{k}(\mathcal{S})$, is given by
\begin{equation}
\vec{L}_{k}(\mathcal{S})=(\eta_j)=(\eta_1, \eta_2, \ldots,\eta_m),
\end{equation}
where $m={n\choose k}$.
\end{definition}
This then gives a map
\begin{equation}
\mathcal{\vec{L}}_k: \mathcal{S}\mapsto(\eta_j),
\end{equation}
where $\mathcal{S}$ is a pre-image of $(\eta_j)$.

We now show that the $k$-RSs of $\mathcal{S}$ is well defined in a sense that it can be characterized by the
maximally mixed state $\rho_M$ of $\mathcal{S}$, which completely specifies $\mathcal{S}$ as it is proportional
to the projection onto $\mathcal{S}$.

\begin{proposition}
Let $\vec{R}_k(\rho_M)=(\gamma_{M,1},\gamma_{M,2},\ldots,\gamma_{M,m})$, which are $k$-RDMs of $\rho_M$. Then
\begin{equation}
\vec{L}_k(\mathcal{S})=(\range(\gamma_{M,1}),\range(\gamma_{M,2}),\ldots,\range(\gamma_{M,m})).
\end{equation}
\end{proposition}
\begin{proof}
Consider any $n$-particle state $\rho$ with
$\range(\rho)\subseteq{\mathcal{S}}$, and let $\vec{R}_k(\rho)=(\gamma_1,\gamma_2,\ldots,\gamma_m)$, which are $k$-RDMs of $\rho$. For each $\gamma_j$, the spectrum decomposition gives $\rho=\sum_k
p_{\alpha}\ket{\psi_{\alpha}}\bra{\psi_{\alpha}}$, with $\ket{\psi_{\alpha}}\in\mathcal{S}$. Let
$\ket{\phi_{\beta}}$ be any orthonormal basis of $n-k$ particles, then
$\gamma_j=\sum_{\alpha,\beta}p_{\alpha}\langle \phi_{\beta}|\psi_{\alpha}\rangle\langle\psi_{\alpha}|\phi_{\beta}\rangle$,
so $\range(\gamma_j)=\Span\{\langle \phi_{\beta}|\psi_{\alpha}\rangle\}$, which is apparently
a subspace of $\range(\gamma_{M,j})$.
\end{proof}

We now give some simple examples for $k$-RSs.
\begin{example}
We consider a three-qubit system with qubits $A,B,C$. For any subspace $\mathcal{S}$ of three qubits, denote $\eta_{1}$ ($\eta_{2}$; $\eta_{3}$) the $2$-RS of the qubits $A,B$ ($B,C$; $A,C$).  Namely, $\vec{L}_k(\mathcal{S})=(\eta_1,\eta_2,\eta_3)$. For the maximally mixed state $\rho_M$ of three qubits, denote $\gamma_{M,1}$ ($\gamma_{M,2}$; $\gamma_{M,3}$) the $2$-RDM of the qubits $A,B$ ($B,C$; $A,C$). We discuss three examples:
\begin{enumerate}
\item $\mathcal{S}$ is one dimensional which is a single state $|001\rangle$. Then $\rho_M=\ket{001}\bra{001}$, $\gamma_{M,1}=\ket{00}\bra{00}$, $\gamma_{M,2}=\gamma_{M,3}=\ket{01}\bra{01}$. Therefore,
\begin{equation}
\vec{L}_2(\mathcal{S})=(\Span\{\ket{00}\},\Span\{\ket{01}\},\Span\{\ket{01}\}).
\end{equation}
\item $\mathcal{S}$ is two dimensional which is spanned by $\{|000\rangle,\ket{111}\}$. Then $\rho_M=\frac{1}{2}(\ket{000}\bra{000}+\ket{111}\bra{111})$, $\gamma_{M,1}=\gamma_{M,2}=\gamma_{M,3}=\frac{1}{2}(\ket{00}\bra{00}+\ket{11}\bra{11})$. Therefore,
\begin{equation}
\vec{L}_2(\mathcal{S})=(\Span\{\ket{00},\ket{11}\},\Span\{\ket{00},\ket{11}\},\Span\{\ket{00},\ket{11}\}).
\end{equation}
\item $\mathcal{S}$ is one dimensional which is a single state $\frac{1}{\sqrt{3}}(|001\rangle+\ket{010}+\ket{100})$. Then $\rho_M=\frac{1}{3}(|001\rangle+\ket{010}+\ket{100})(\bra{001}+\bra{010}+\bra{100})$, $\gamma_{M,1}=\gamma_{M,2}=\gamma_{M,3}=\frac{1}{3}\ket{00}\bra{00}+\frac{2}{3}(\ket{01}+\ket{10})(\bra{01}+\bra{10})$. Therefore,
\begin{eqnarray}
&&\vec{L}_2(\mathcal{S})=(\Span\{\ket{00},\frac{1}{\sqrt{2}}(\ket{01}+\ket{10})\},\nonumber\\
&&\Span\{\ket{00},\frac{1}{\sqrt{2}}(\ket{01}+\ket{10})\},\Span\{\ket{00},\frac{1}{\sqrt{2}}(\ket{01}+\ket{10})\}).
\end{eqnarray}
\end{enumerate}
\end{example}

We denote the set of all the $k$-RSs by $\Theta_k =
\{\vec{L}_k(\mathcal{S}) \,|\, \mathcal{S}\in\Theta\}$. Note that
the map $\mathcal{\vec{L}}_k$, similar to $\mathcal{\vec{R}}_k$, is
not one-to-one. That is, there may exist some other
$\mathcal{S}'\in\Theta$ such that
$\vec{L}_{k}(\mathcal{S}')=\vec{L}_{k}(\mathcal{S})$. Consequently,
the set of pre-images of $\vec{L}_{k}(\mathcal{S})$ contains
$\mathcal{S}$. For the convenience of later discussion, for any
$\vec{L}_{k}$ we define a special pre-image as follows.
\begin{definition}
For any $\vec{L}_{k}=(\eta_1,\ldots,\eta_m)\in\Theta_k$, the MPI of
$\vec{L}_{k}$ is a pre-image of $\vec{L}_{k}$ under the map
$\mathcal{\vec{L}}_k$, which is given by $\bigcap_j \eta_j\otimes
\iota_{\bar{j}}$, where $\iota_{\bar{j}}$ is the Hilbert space of
the $n-k$ particles that  $\eta_j$ does not act on.
\end{definition}

Apparently, we have
\begin{lemma}
For any $\vec{L}_{k}=(\eta_1,\ldots,\eta_m)\in\Theta_k$, the MPI is
its maximal pre-image under the map $\mathcal{\vec{L}}_k$.
\end{lemma}

\section{The Semigroup Structure of $\Theta_k$}
\label{sec:semigroup}

To understand the mathematical structure of $\Theta_k$,
we define a binary operation, called sum, denoted by $+$, for any two elements in $\Theta_k$.
\begin{definition}
For $\mathcal{S}_1,\mathcal{S}_2\in\Theta$, let $\vec{L}_k(\mathcal{S}_1)=(\eta_1,\eta_2,\ldots,\eta_m)$ and
$\vec{L}_k(\mathcal{S}_2)=(\xi_1,\xi_2,\ldots,\xi_m)$, define
\begin{equation}
\vec{L}_k(\mathcal{S}_1)+\vec{L}_k(\mathcal{S}_2)=(\eta_1+\xi_1,\eta_2+\xi_2,\ldots,\eta_m+\xi_m).
\end{equation}
\end{definition}

\begin{lemma}
\label{lm:pre} The MPI of
$\vec{L}_k(\mathcal{S}_1)+\vec{L}_k(\mathcal{S}_2)$ under the map
$\mathcal{\vec{L}}_k$ contains the sum of the MPI of
$\vec{L}_k(\mathcal{S}_1)$ and the MPI of
$\vec{L}_k(\mathcal{S}_2)$.
\end{lemma}
\begin{proof}
Recall that for any vector
$\vec{L}_k=(\eta_1,\eta_2,\ldots,\eta_m)\in\Theta_k$, the MPI under
the map $\mathcal{\vec{L}}_k$ is given by $\bigcap_j \eta_j\otimes
\iota_{\bar{j}}$. Let
$\vec{L}_k(\mathcal{S}_1)=(\eta_1,\eta_2,\ldots,\eta_m)$,
$\vec{L}_k(\mathcal{S}_2)=(\xi_1,\xi_2,\ldots,\xi_m)$, then one has
\begin{equation}
\bigcap_j(\eta_j\otimes \iota_{\bar{j}}+\xi_j\otimes \iota_{\bar{j}})\supseteq \bigcap_j(\eta_j\otimes \iota_{\bar{j}})+\bigcap_j(\xi_j\otimes \iota_{\bar{j}}),
\end{equation}
where the left hand side is the the MPI of $\vec{L}_k(\mathcal{S}_1)+\vec{L}_k(\mathcal{S}_2)$,
and the right hand side is the sum of the MPI of $\vec{L}_k(\mathcal{S}_1)$ and the MPI of $\vec{L}_k(\mathcal{S}_2)$. Note that any state $\ket{\phi}$ in the space given by the right hand side can be written as $\ket{\phi}=\ket{\phi_{\alpha}}+\ket{\phi_{\beta}}$ where $\ket{\phi_{\alpha}}\in \bigcap_j (\eta_j\otimes \iota_{\bar{j}})$ and $\ket{\phi_{\beta}} \in \bigcap_j (\xi_j \otimes \iota_{\bar{j}})$. Therefore $\ket{\phi_{\alpha}}\in\eta_j\otimes \iota_{\bar{j}}$ and $\ket{\phi_{\beta}}\in\xi_j \otimes \iota_{\bar{j}}$ for any $j$. Then  $\ket{\phi}=\ket{\phi_{\alpha}}+\ket{\phi_{\beta}}\in\eta_j\otimes \iota_{\bar{j}}+\xi_j\otimes \iota_{\bar{j}}$ for any $j$, which then follows the inclusion.
\end{proof}

We examine an example of Lemma~\ref{lm:pre}.
\begin{example}
We consider a three-qubit system with qubits $A,B,C$. For the three-qubit  subspace $\mathcal{S}_1=\Span\{\frac{1}{\sqrt{3}}(\ket{001}+\ket{010}+\ket{100})\}$ and $\mathcal{S}_2=\Span\{\ket{111}\}$, we have
\begin{eqnarray}
&&\vec{L}_2(\mathcal{S}_1)=(\Span\{\ket{00},\frac{1}{\sqrt{2}}(\ket{01}+\ket{10})\},\nonumber\\
&&\Span\{\ket{00},\frac{1}{\sqrt{2}}(\ket{01}+\ket{10})\},\Span\{\ket{00},\frac{1}{\sqrt{2}}(\ket{01}+\ket{10})\}),
\end{eqnarray}
and
\begin{equation}
\vec{L}_2(\mathcal{S}_2)=(\Span\{\ket{11}\},\Span\{\ket{11}\},\Span\{\ket{11}\}).
\end{equation}
Therefore,
\begin{eqnarray}
&&\vec{L}_2(\mathcal{S}_1)+\vec{L}_2(\mathcal{S}_2)=(\Span\{\ket{00},\frac{1}{\sqrt{2}}(\ket{01}+\ket{10}),\ket{11}\},\nonumber\\
&&\Span\{\ket{00},\frac{1}{\sqrt{2}}(\ket{01}+\ket{10}),\ket{11}\},\Span\{\ket{00},\frac{1}{\sqrt{2}}(\ket{01}+\ket{10}),\ket{11}\}).
\end{eqnarray}
The MPI of $\vec{L}_2(\mathcal{S}_1)+\vec{L}_2(\mathcal{S}_2)$ is then the three-qubit symmetric space spanned by
\begin{equation}
\{\ket{000},\frac{1}{\sqrt{3}}(\ket{001}+\ket{010}+\ket{100}),\frac{1}{\sqrt{3}}(\ket{110}+\ket{101}+\ket{011}),\ket{111}\},
\end{equation}
which contains the sum
of the MPI of $\vec{L}_2(\mathcal{S}_1)$ ($\Span\{\ket{000},\frac{1}{\sqrt{3}}(\ket{001}+\ket{010}+\ket{100})\}$) and
the MPI of $\vec{L}_2(\mathcal{S}_2)$ ($\Span\{\ket{111}\}$).
\end{example}

We are now ready to establish the semigroup structure of $\Theta_k$.

\begin{theorem}
\label{th:semigroup}
With the binary operation sum, the set of all $k$-RSs $\Theta_k$ form an idempotent, commutative semigroup without zero element.
\end{theorem}
\begin{proof}
First of all, we show the set $\Theta_k$ is closed under the sum operation. This is a direct consequence of Lemma~\ref{lm:pre}, as the
MPI of $\vec{L}_k(\mathcal{S}_1)$ contains $\mathcal{S}_1$ and the
MPI of $\vec{L}_k(\mathcal{S}_2)$ contains $\mathcal{S}_2$. Therefore, the MPI of
$\vec{L}_k(\mathcal{S}_1)+\vec{L}_k(\mathcal{S}_2)$ contains $\mathcal{S}_1+\mathcal{S}_2$, which is a non-empty
$n$-particle subspace, in this sense $\Theta_k$ is a semigroup. $\Theta_k$ is idempotent given
$\vec{L}_k(\mathcal{S}_1)+\vec{L}_k(\mathcal{S}_1)=\vec{L}_k(\mathcal{S}_1)$, and commutative given
$\vec{L}_k(\mathcal{S}_1)+\vec{L}_k(\mathcal{S}_2)=\vec{L}_k(\mathcal{S}_2)+\vec{L}_k(\mathcal{S}_1)$.

For the commutative semigroup $\Theta_k$ whose group operation is
$+$, the identity element is usually called zero element, denote by
$\vec{O}_k$, which satisfies for any $\vec{L}_k\in\Theta_k$,
$\vec{L}_k+\vec{O}_k=\vec{O}_k+\vec{L}_k=\vec{L}_k$. However, we
show that there is no such a zero element $\vec{O}_k$ in $\Theta_k$.
If there does exist such an $\vec{O}_k=(o_1,o_2,\ldots,o_m)$, where
each $o_i$ must be the empty space of $k$-particles, then the
pre-image of $\vec{O}_k$ is nothing but the empty space $O$ of the
$n$-particle space. Remember that the purpose that we introduce
$\Theta_k$ is to study the ground-state space structure of
frustration-free hamiltonians, so we suppose the space involved is
at least one-dimensional. Thus, in $\Theta_k$ such a zero element
does not exist.
\end{proof}

\section{The Semilattice Structure of $\Theta_k$}
\label{sec:semilattice}

It turns out that this special kind of semigroup discussed in
Theorem~\ref{th:semigroup}, which is idempotent and commutative, is
a structure called semilattice, which is widely studied in order
theory~\cite{Kal83}.  The binary operation $+$ is usually called
`join', with a notation $\vee$. So the semigroup $\Theta_k$ is then
a join-semilattice. A formal definition of join-semilattice is given
below, then one can readily check for $\Theta_k$.
\begin{definition}
\label{def:joinsemi}
A join-semilattice is an algebraic structure $<S,\vee>$ consisting of a set $S$ with a binary
operation $\vee$, called join, such that for all members $x, y$, and $z$ of $S$, the following identities hold:
\begin{enumerate}
\item Associativity: $x \vee (y \vee z) = (x \vee y) \vee z$.
\item Commutativity: $x \vee y = y \vee x$.
\item Idempotency: $x \vee x = x$.
\end{enumerate}
\end{definition}
Equivalently, there is another order-theoretic definition of join-semilattice~\cite{Kal83}.
\begin{definition}
A set $S$ partially ordered by the binary relation $\leq$ is a join-semilattice if
for all elements $x$ and $y$ of $S$, the least upper bound of the set ${x, y}$ exists.
The least upper bound of the set ${x, y}$ is called the join of $x$ and $y$, denoted
by $x \vee y$.
\end{definition}
For $\Theta_k$, the binary relation $\leq$ is actually the set
inclusion $\subseteq$. More precisely, for two elements
$\vec{L}_k=(\eta_1,\eta_2,\ldots,\eta_m)$ and
$\vec{L}'_k=(\xi_1,\xi_2,\ldots,\xi_m)$ in $\Theta_k$,
$\vec{L}_k\leq\vec{L}'_k$ if $\eta_j\subseteq\xi_j$ for $\forall\
j$. Moreover, we say $\vec{L}_k<\vec{L}'_k$ if
$\eta_j\subseteq\xi_j$ for $\forall\ j$ and for at least one $j$,
$\eta_j\subset\xi_j$.

Therefore, in order-theoretic terms, $\Theta_k$ is a
join-semilattice without zero element. This allows us to investigate
the structure of $\Theta_k$ within the order-theoretic framework.
Let us first mention that there are some important order-theoretic
notations regarding join-semilaltice~\cite{Kal83}.
\begin{definition}
$x$ is an atom if there exists no nonzero element $y$ of $L$ such that $y < x$.
\end{definition}
\begin{definition}
$x$ is a join irreducible iff $x = a\vee b$ implies $x = a$ or $x = b$ for any $a,b$ in $S$.
\end{definition}
Any atom is also a join irreducible, however the reverse is generally not true.
\begin{definition}
$x$ is a join prime iff $x \leq a \vee b$ implies $x \leq a$ or $x \leq b$.
\end{definition}
Any join prime element is also a join irreducible, however the reverse might not be true.

Now we study in more detail the semilattice structure of $\Theta_k$, by examining the property of atoms, join primes and join irreducibles. We start from atoms.

Intuitively, the name `atom', borrowing from physics, means the
`basic' elements (or building blocks) of $\Theta_k$. Indeed, by
definition, atoms are smallest elements of  $\Theta_k$. For any
$\vec{L}_k(\mathcal{S})=(\eta_1,\eta_2,\ldots,\eta_3)\in\Theta_k$,
if each $\eta_j$ is only of dimension one, then
$\vec{L}_k(\mathcal{S})$ is an atom, and the pre-image of
$\vec{L}_k(\mathcal{S})$ is a product state of $n$-particles. In
other words, $k$-RSs corresponding to a single product state are
atoms. However, atoms can be much more complicated, that is, not all
the atoms correspond to single product states. We have the following
proposition.
\begin{proposition}
\label{pro:pure}
If the MPI of an element $\vec{L}_k\in\Theta_k$ is a single pure state, then $\vec{L}_k$ is an atom.
\end{proposition}
\begin{proof}
Suppose $\vec{L}_k$ is not an atom, then there exists
$\vec{L}'_k\in\Theta_k$ such that $\vec{L}'_k<\vec{L}_k$.
Consequently, the MPI of $\vec{L}'_k$ under the map
$\mathcal{\vec{L}}_k$ is contained in the pre-image of $\vec{L}_k$.
However, the MPI of $\vec{L}_k$ is a single pure state, so the MPI
of $\vec{L}'_k$ can only be this pure state. Therefore
$\vec{L}'_k=\vec{L}_k$, which gives that $\vec{L}_k$ is an atom of
$\Theta_k$.
\end{proof}

We provide an example for Proposition~\ref{pro:pure}.
\begin{example}
Consider a three-qutrit system with qutrits $A,B,C$. Choose the basis vectors of the Hilbert space for each qutrit to be $\ket{0},\ket{1},\ket{2}$.
Consider the following $\vec{L}_2\in\Theta_2$.
\begin{eqnarray}
&&\vec{L}_2=(\Span\{\sqrt{\frac{2}{3}}(\ket{00}-\frac{1}{\sqrt{2}}\ket{12})\sqrt{\frac{2}{3}}(\ket{11}-\frac{1}{\sqrt{2}}\ket{02})\},\nonumber\\
&&\Span\{\sqrt{\frac{2}{3}}(\ket{00}-\frac{1}{\sqrt{2}}\ket{21}),\sqrt{\frac{2}{3}}(\ket{11}-\frac{1}{\sqrt{2}}\ket{20})\},\nonumber\\
&&\Span\{\ket{00},\frac{1}{\sqrt{2}}(\ket{01}+\ket{10}),\ket{11}\}).
\end{eqnarray}
The MPI of $\vec{L}_2$ under the map $\mathcal{\vec{L}}_k$ is a
single pure state $\ket{\psi}$ which is given by
\begin{equation}
\ket{\psi}=\frac{1}{\sqrt{3}}\{\ket{000}-\frac{1}{\sqrt{2}}(\ket{021}+\ket{120})+\ket{111}\}.
\end{equation}
Apparently $\ket{\psi}$ is not a product state of three qutrits.
\end{example}

Now we turn to join primes of $\Theta_k$. Similarly as for the atom case, we start to discuss the case for $k$-RSs whose MPI is a single product state. Unlike in the atom case, and indeed quite counterintuitively, we observe the following lemma.
\begin{lemma}
In $\Theta_k$, $k$-RSs whose MPI is a single product state is not a join prime.
\end{lemma}
\begin{proof}
Consider an $n$-particle system. Up to local unitary equivalence, any product state can be written as $\ket{\psi_0}=\ket{0}\ket{0}\ldots\ket{0}$.
Now consider other two states $\ket{\psi_1}=\frac{1}{\sqrt{2}}(\ket{00}+\ket{11})\ket{0}\ldots\ket{0}$ and
$\ket{\psi_2}=\ket{0}\ldots\ket{0}\frac{1}{\sqrt{2}}(\ket{00}+\ket{11})$. It is straightforward to check that
\begin{equation}
\vec{L}_2(\Span\{\ket{\psi_1}\})\vee\vec{L}_2(\Span\{\ket{\psi_2}\})\geq\vec{L}_2(\Span\{\ket{\psi_0}\}).
\end{equation}
That is, the join of the $2$-RSs of $\ket{\psi_1}$
and $\ket{\psi_2}$ contains the $2$-RSs of $\ket{\psi_0}$. However, apparently $\vec{L}_2(\Span\{\ket{\psi_0}\})\nleq \vec{L}_2(\Span\{\ket{\psi_1}\})$
and $\vec{L}_2(\Span\{\ket{\psi_0}\})\nleq \vec{L}_2(\Span\{\ket{\psi_2}\})$. So $2$-RSs whose MPI is a single
product state is not a join prime. Similar idea applies to $k>2$, therefore $k$-RSs whose MPI is a single product state is not join prime.
\end{proof}

As we have already demonstrated, the $k$-RSs whose MPI is a single
product state is a `minimum' element in $\Theta_k$ (atoms). However,
they are not join primes, which hints the following important structure property of $\Theta_k$.
\begin{theorem}
There is no join prime in $\Theta_k$.
\end{theorem}
\begin{proof}
One just needs to prove that for any non-product state $\ket{\psi_0}$, the $k$-RSs $\vec{L}_k(\Span(\ket{\psi_0}))$
is not join prime, as the case of $\vec{L}_k(\mathcal{S})$ for any $\mathcal{S}\in\Theta$ follows a similar proof.

Since $\ket{\psi_0}$ is not a product state, there exists a particle
$\alpha$ which is entangled with the other $n-1$ particles. Without
loss of generality, we choose that $\alpha$ is just the first
particle. For $\Span(\ket{\psi_0})$, consider the maximally mixed
state of the $(n-1)$-RS of particles $\{2,3,\ldots, n\}$, and denote
it by $\rho_{\bar{\alpha}}$. And denote the maximally mixed state of
the  $1$-RS of the first particle particle by $\rho_{\alpha}$. We
know $\rho_{\alpha}$ is at least of rank 2, and let us write
$\rho_{\alpha}$ as
$\rho_{\alpha}=\frac{1}{r}\sum_{i=1}^{r}\ket{\alpha_i}\bra{\alpha_i}=\rho_a+\rho_b$,
where $\rho_a=\sum_{i=1}^{m}\ket{\alpha_i}\bra{\alpha_i}$
$\rho_b=\sum_{i=m+1}^{r}\ket{\alpha_i}\bra{\alpha_i}$ for some
$1\leq m\leq r$. Then it is straightforward to show that
\begin{equation}
\vec{L}_k(\range(\rho_a\otimes\rho_{\bar{\alpha}}))\vee \vec{L}_k(\range(\rho_b\otimes\rho_{\bar{\alpha}}))\geq \vec{L}_k(\Span(\ket{\psi_0}))
\end{equation}
That is, the
join of the $k$-RSs of $\range(\rho_a\otimes\rho_{\bar{\alpha}})$
and the $k$-RSs of $\range(\rho_b\otimes\rho_{\bar{\alpha}})$ contain the $k$-RSs of $\Span(\ket{\psi_0})$.
However, apparently $\vec{L}_k(\Span(\ket{\psi_0}))\nleq \vec{L}_k(\range(\rho_a\otimes\rho_{\bar{\alpha}}))$
and $\vec{L}_k(\Span(\ket{\psi_0}))\nleq \vec{L}_k(\range(\rho_b\otimes\rho_{\bar{\alpha}}))$.
Therefore the $k$-RSs $\vec{L}_k(\Span(\ket{\psi_0}))$ is not join prime.
\end{proof}

Finally we discuss join irreducibles of $\Theta_k$. We know that any atom is also a join irreducible, however the reverse is generally not true. This is indeed the case for $\Theta_k$, as shown by the following example.
\begin{example}
\label{eq:irr}
Consider a three-qubit system with qubits $A,B,C$, and the three qubit $W$ state $\ket{W}=\frac{1}{\sqrt{3}}(\ket{001}+\ket{010}+\ket{100})$.
Then consider the following $\vec{L}_2(W)=\vec{L}_2(\Span(\ket{W}))\in\Theta_2$.
\begin{eqnarray}
&&\vec{L}_2(W)=(\Span\{\ket{00},\frac{1}{\sqrt{2}}(\ket{01}+\ket{10})\},\nonumber\\
&&\Span\{\ket{00},\frac{1}{\sqrt{2}}(\ket{01}+\ket{10})\},\Span\{\ket{00},\frac{1}{\sqrt{2}}(\ket{01}+\ket{10})\}).
\end{eqnarray}
We now show that $\vec{L}_2(W)$ is a join irreducible. If not, then there exist $\vec{L}'_2,\vec{L}''_2\in\Theta_2$ such that
$\vec{L}_2(W)=\vec{L}'_2\vee\vec{L}''_2$ but $\vec{L}_2\neq\vec{L}'_2$ and $\vec{L}_2\neq\vec{L}''_2$.
Then $\vec{L}'_2$ and $\vec{L}''_2$ must be of the forms
\begin{eqnarray}
&&\vec{L}'_2=(\Span\{\alpha_{1}'\ket{00}+\beta'_{1}\frac{1}{\sqrt{2}}(\ket{01}+\ket{10})\},\nonumber\\
&&\Span\{\alpha_{2}'\ket{00}+\beta'_{2}\frac{1}{\sqrt{2}}(\ket{01}+\ket{10})\},\nonumber\\
&&\Span\{\alpha_{3}'\ket{00}+\beta'_{3}\frac{1}{\sqrt{2}}(\ket{01}+\ket{10})\}).
\end{eqnarray}
and
\begin{eqnarray}
&&\vec{L}''_2=(\Span\{\alpha_{1}''\ket{00}+\beta''_{1}\frac{1}{\sqrt{2}}(\ket{01}+\ket{10})\},\nonumber\\
&&\Span\{\alpha_{2}''\ket{00}+\beta''_{2}\frac{1}{\sqrt{2}}(\ket{01}+\ket{10})\},\nonumber\\
&&\Span\{\alpha''_{3}\ket{00}+\beta''_{3}\frac{1}{\sqrt{2}}(\ket{01}+\ket{10})\}).
\end{eqnarray}
However, it is straightforward to check if $\vec{L}'_2,\vec{L}''_2\in\Theta_2$, then one must have $\beta'_{1}=\beta''_{1}=\beta'_{2}=\beta''_{2}=\beta'_{3}=\beta''_{3}=0$. Therefore $\vec{L}_2$ is a join irreducible.

To show that $\vec{L}_2(W)$ is not an atom, note that the atom
\begin{equation}
(\Span\{\ket{00}\},\Span\{\ket{00}\},\Span\{\ket{00}\}).
\end{equation}
is contained in $\vec{L}_2(W)$.
\end{example}

Note that this example also shows that $\vec{L}_2(W)$ is not a join of atoms, meaning that elements in
$\Theta_k$ may not be join of atoms. Elements in $\Theta_k$ which are not join irreducibles can be written as join of other two elements. Then a natural question is whether
any element in $\Theta_k$ can be written as finite join of join irreducibles. This question can indeed be answered by borrowing some general theory of semilattices.
We first need the concept of chains. For any partially ordered set, a chain is just any totally ordered subset. Then we need
to the following `descending chain condition'.
\begin{definition}
A partially ordered set $S$ satisfies the descending chain condition if there does not exist an infinite descending chain $s_1>s_2,\ldots$ of elements of $S$.
\end{definition}
\begin{theorem}~\cite{DP02,Cohn02}
\label{th:str}
If a semilattice satisfies the descending chain condition, then every element can be expressed as a finite join of join irreducible elements.
\end{theorem}
\begin{corollary}
\label{cor:str}
Every element in $\Theta_k$ can be expressed as a finite join of join irreducible elements.
\end{corollary}
\begin{proof}
For $\Theta_k$ with respect to an $n$-particle system whose Hilbert
space is finite-dimensional, elements in $\Theta_k$ are just finite
dimensional  vector spaces, where one could not have any infinite
descending chains of subspaces. Therefore the semilattice $\Theta_k$
satisfies the descending chain condition. Consequently,
Theorem~\ref{th:str} applies.
\end{proof}

It should be pointed out that the characterizations of atoms and
join irreducible elements in $\Theta_k$ need further investigation.
For instance, the pre-images of atoms under the map
$\mathcal{\vec{L}}_k$ seem interesting, and actually we conjecture
that each of them always corresponds to one single point in $\D_k$,
but this has not been proved. Another example is related to the gap
between atoms and join irreducible elements. In
Example~\ref{eq:irr}, we demonstrated that one join irreducible
element is possible not to be an atom. Instead, it has a
two-dimensional pre-image, and one subspace of this pre-image
corresponds to a real atom. About this example, two natural
questions that might be asked are: for a non-atom join irreducible
element, must the pre-image under the map $\mathcal{\vec{L}}_k$ be
two-dimensional? Can we find a non-atom join irreducible element
such that it is bigger than two or more elements? In a word, the
semilattice structure of $\Theta_k$ needs more exploration.

\section{The Convex Structure of $\Theta_k$}
\label{sec:convex}

In Sec.~\ref{sec:semilattice}, we have analyzed the semilattice structure of $\Theta_k$ in detail. Given the discussion in Sec.~\ref{sec:Dk} that the set of all $k$-RDMs, namely $\D_k$, is a convex set, it will be nice if there could be a way to view the semilattice structure of $\Theta_k$ as some kind of convex structure. Indeed,
there is a general theory of convex structure, called the theory of convex spaces~\cite{Fri09}. This theory unifies aspects of convex geometry with aspects of order theory (i.e. semilattices). It has two extreme cases: 1. convex spaces of geometric type, which is related to the structure of $\D_k$; 2. convex spaces of combinatorial type, which is related to the structure of $\Theta_k$. We start from the definition of convex spaces.
\begin{definition}
~\cite{Fri09,Fri09a} A convex space is a set $C$ together with a
family of binary operations
\begin{equation}
cc_{\alpha}:\ C\times C\rightarrow C,\ \alpha\in (0,1)
\end{equation}
satisfying the the conditions of:
\begin{enumerate}
\item Idempotency: $cc_{\alpha}(x,x)=x$.
\item Parametric commutativity: $cc_{\alpha}(x,y)=cc_{1-\alpha}(y,x)$.
\item Deformed parametric associativity: $cc_{\alpha}(cc_{\mu}(x,y),z)=cc_{\tilde{\alpha}}(x,cc_{\tilde{\mu}}(y,z))$,
\end{enumerate}
where
\begin{equation}
\widetilde{\alpha}=\alpha\mu,\qquad\widetilde{\mu}=\left\{\begin{array}{cl}\frac{\alpha(1-\mu)}{\:\stackrel{}{(1-\alpha\mu)}\:}&\textrm{ if }\alpha\mu\neq 1\\\textrm{arbitrary}&\textrm{ if }\alpha=\mu=1.\end{array}\right.
\end{equation}
\end{definition}
One can write $cc_{\alpha}=\alpha x+(1-\alpha)y$ as the usual notation for convex
combinations (the convex space of geometric type). The convex space of combinatorial type is defined as follows.
\begin{definition}
~\cite{Fri09,Fri09a}
\label{df:com}
A convex space $C$ is said to be of combinatorial type whenever all convex combinations
\begin{equation}
\alpha x+(1-\alpha)y
\end{equation}
is independent of $\alpha$.
\end{definition}

It turns out~\cite{Fri09a}, a convex space of combinatorial type is nothing but a set $C$ together with a single binary operation $\vee:\ C\times C\rightarrow C$, which is idempotent, commutative, and associative.  By defining an order structure $x\geq y\ \leftrightarrow\ x=x\vee y$, it can be directly verified that such a $C$ is nothing but a join-semilattice. Therefore, in the language of abstract convexity, $\Theta_k$ is a convex space of combinatorial type.

Given Definition~\ref{df:com} one can define extreme points for convex spaces
of combinatorial type, similar as the extremes points for the usual convex structure.
\begin{definition}
For a convex space $C$ of combinatorial type, $x\in C$ is an extreme
point if there does not exist two different non-zero elements
$y,z\in C$ such that $x=y+z$.
\end{definition}
In the order theory language, for a semilattice $L$, $x\in L$
corresponds to an extreme point if there does not exists non-zero
elements $y,z\in L$ such that $x=y\vee z$. It is straightforward to
see that these extreme points are nothing but atoms in $L$, as if
$x$ is an atom, then $x\neq a\vee b$ for any $a,b\in L$, which means
$x$ is an extreme point; if $x$ is not an atom, then an element $x'$
different from $x$ exists such that $x=x\vee x'$, thus $x$ is not an
extreme point.

That is, in the language of abstract convexity, atoms are proper analogs of extreme points. Similar as the study
of extreme points in $\mathcal{D}_k$, it is important to discuss the properties of atoms of $\Theta_k$.
Besides the properties we have already discussed in Sec.~\ref{sec:semilattice}, here we would like to further examine one more aspect
regarding atoms of $\Theta_k$.
We know that any point in $\mathcal{D}_k$ can always be expressed as a weighted sum of extreme points of $\mathcal{D}_k$.
That is, extreme points are `building blocks' of $\mathcal{D}_k$. However, this is not true for atoms in $\Theta_k$. That is,
there exists elements in $\Theta_k$ that cannot be expressed as a sum of atoms. An example is provided in Example~\ref{eq:irr}, where $\vec{L}_2(W)\in\Theta_2$ is such an element. This kind of examples indicate that the structure of $\Theta_k$ is in a sense richer than the structure of $\mathcal{D}_k$. Understanding the atoms of $\Theta_k$, although very important to the understanding of the structure of $\Theta_k$, is still not enough to fully characterize $\Theta_k$. Indeed, as it is shown in Corollary~\ref{cor:str} that every element in $\Theta_k$ can be expressed as a finite join of join irreducible elements, one will also need to understand the properties of join irreducibles of $\Theta_k$.

\section{Ground-State Spaces of Frustration-Free Hamiltonians}
\label{sec:GSS}

In Sec.~\ref{sec:Dk}, we discussed the relationship regarding the
geometry of $\mathcal{D}_k$ and the ground-state spaces of $k$-local
Hamiltonians, obtained by considering the dual core of
$\mathcal{D}_k$, which are nothing but the set of all the $k$-local
Hamiltonians. To relate the structure of $\Theta_k$ to the
ground-state spaces of $k$-local FF Hamiltonians, we will also
introduce a dual theory. We show that the set of all $k$-local FF
Hamiltonians form a `dual semilattice' of $\Theta_k$, called a
meet-semilattice.

A $k$-local Hamiltonian $H=\sum_j H_j$ with each $H_j$ acting on at
most $k$-particles, is FF if the ground states of $H$ are also the
ground states of each $H_j$. In this paper, we only focus on the
ground-state space properties of $H_j$, but not any excited state
properties. In this case, one can replace each $H_j$ with a
projector $\Pi_j$, where $\Pi_j$ is the projection onto the
orthogonal space of the ground-state space of $H_j$~\cite{Bra06}. We
then put the total $m={n\choose k}$ $k$-particle projectors in a
fixed order and write $H$ as a vector
\begin{equation}
\vec{H}_k=(\Pi_1,\Pi_2,\ldots,\Pi_m),
\end{equation}
where the subscript $k$ means that $H$ is $k$-local. The
ground-state space of $\vec{H}_k$ is then given by
\begin{equation}
\bigcap_j(\Pi_j^{\perp}\otimes \iota_{\bar{j}}),
\end{equation}
where $\Pi_j^{\perp}$ is the ground-state space of $\Pi_j$.

Note that even if $\vec{H}_k$ is frustration-free, it is not necessarily true that the vector $\vec{H}^{\perp}_k=(\Pi^{\perp}_1,\Pi^{\perp}_2,\ldots,\Pi^{\perp}_m)\in\Theta_k$.
However in the rest of the paper, we only consider the case where $\vec{H}^{\perp}_k\in\Theta_k$. Indeed,
for any $\vec{L}_k=(\eta_1,\eta_2,\dots,\eta_m)\in\Theta_k$, $\vec{L}^{\perp}_k=(\eta^{\perp}_1,\eta^{\perp}_2,\dots,\eta^{\perp}_m)$
is a $k$-local FF Hamiltonian. Here $\eta^{\perp}_i$ is interpreted as a projection, which is just the projection onto the orthogonal space of $\eta_i$. In this sense, $k$-local FF Hamiltonians are duals of the elements in $\Theta_k$. We discuss an example.
\begin{example}
\label{eq:H}
Consider a three-qubit system with qubits $A,B,C$, and a $2$-local FF Hamiltonian $\vec{H}_2$ with
\begin{equation}
\vec{H}^{\perp}_2=(\Span\{\ket{00},\frac{1}{\sqrt{2}}(\ket{01}+\ket{10})\},\Span\{\ket{00}\},\Span\{\ket{00}\}).
\end{equation}
$\vec{H}^{\perp}_2$ is not in $\Theta_2$ as there does not exist any
subspace of three qubits whose image under the map
$\mathcal{\vec{L}}_2$ gives $\vec{H}^{\perp}_2$. However, the ground
space of $\vec{H}_2$ is a single product state $\ket{000}$, whose
image under the map $\mathcal{\vec{L}}_2$ gives
\begin{equation}
\vec{L}_2(\ket{000})=(\Span\{\ket{00}\},\Span\{\ket{00}\},\Span\{\ket{00}\})\in\Theta_2,
\end{equation}
so one can replace $\vec{H}_2$ by $\vec{L}^{\perp}_2(\ket{000})$, which has the same ground-state space as that of $\vec{H}_2$.
\end{example}

In general, if $\vec{H}_k$ is FF but $\vec{H}^{\perp}_k$ is not in
$\Theta_k$, then there does not exist any $n$-particle state whose
image under the map $\mathcal{\vec{L}}_k$ gives $\vec{H}^{\perp}_k$.
However, take the image of the ground-state space of $\vec{H}_k$
under the map $\mathcal{\vec{L}}_k$ will result in an element
$\vec{L}_k$ in $\Theta_k$, such that $\vec{L}_k\leq
\vec{H}^{\perp}_k$. Therefore the ground-state space of $\vec{H}_k$
is also the ground-state space of $\vec{L}^{\perp}_k$. In this
sense, one can just replace $\vec{H}_k$ by $\vec{L}^{\perp}_k$ for
our discussion of ground-state space of $k$-local FF Hamiltonians.

Now for an $n$-particle system, define a set of the $k$-local FF Hamiltonians:
\begin{equation}
\Upsilon_k=\{\vec{H}_k|\vec{H}^{\perp}_k\in\Theta_k\}.
\end{equation}
One can define a binary operation $\wedge$ for any $\vec{H}_k,\vec{H}'_k\in\Upsilon_k$:
\begin{equation}
\vec{H}_k\wedge\vec{H}'_k=((\Pi^{\perp}_1\cap\Pi'^{\perp}_1)^{\perp},(\Pi^{\perp}_2\cap\Pi'^{\perp}_2)^{\perp},\ldots,(\Pi^{\perp}_m\cap\Pi'^{\perp}_m)^{\perp}),
\end{equation}
where the operation $\cap$ on two vector spaces is the usual
intersection of vector spaces, and
$(\Pi^{\perp}_1\cap\Pi'^{\perp}_1)^{\perp}$ is the projection with
ground-state space being $\Pi^{\perp}_1\cap\Pi'^{\perp}_1$. It is
then straightforward to check that with the binary operation
$\wedge$, $\Upsilon_k$ form a meet-semilattice. For the definition
of meet-semilattice, one can just replace the binary operation
$\vee$ by $\wedge$ in the definition of join-semilattice given by
Definition~\ref{def:joinsemi}.

Given this duality between elements in $\Theta_k$ and $k$-local FF
Hamiltonians, we are ready to discuss the relationship between
elements in $\Theta_k$ and ground-state spaces of $k$-local FF
Hamiltonians. Recall that for an $n$-particle subspace
$\mathcal{S}$, in general the pre-image of $\vec{L}_k(\mathcal{S})$
under the map $\mathcal{\vec{L}}_k$ contains $\mathcal{S}$. We then
have the following proposition when they are equal.
\begin{proposition}
The MPI of $\vec{L}_k(\mathcal{S})$ under the map
$\mathcal{\vec{L}}_k$ equals $\mathcal{S}$ if and only if
$\mathcal{S}$ is the ground-state space of a $k$-local FF
Hamiltonian.
\end{proposition}
\begin{proof}
This can be readily proved by the duality property. First of all,
the MPI of any element
$\vec{L}_k=(\eta_1,\eta_2,\dots,\eta_m)\in\Theta_k$ is obviously a
ground-state space of a $k$-local FF Hamiltonians $\vec{H}_k$, and
the Hamiltonian is given by
$\vec{H}_k=(\eta^{\perp}_1,\eta^{\perp}_2,\dots,\eta^{\perp}_m)$ as
discussed above. Therefore, if the MPI of $\vec{L}_k(\mathcal{S})$
under the map $\mathcal{\vec{L}}_k$ equals $\mathcal{S}$, then
$\mathcal{S}$ is the ground-state space of the $k$-local FF
Hamiltonians $\vec{H}_k$.

For the reverse, if $\mathcal{S}$ is a ground-state space of a $k$-local FF Hamiltonian $\vec{H}_k=(\Pi_1,\Pi_2,\ldots,\Pi_m)\in\Upsilon_k$, then $\vec{H}^{\perp}_k=(\Pi^{\perp}_1,\Pi^{\perp}_2,\ldots,\Pi^{\perp}_m)\in\Theta_k$, and the MPI of $\vec{H}^{\perp}_k$ is just $\mathcal{S}$.
\end{proof}

Now we further relate ground-state spaces of $k$-local FF Hamiltonians to subsemilattices of $\Theta_k$.
\begin{proposition}
\label{pro:sub} The MPIs of a subsemilattice of $\Theta_k$ under the
map $\mathcal{\vec{L}}_k$ is a ground-state space of some $k$-local
FF Hamiltonian. The reverse is also true, that is, all the subspaces
of a ground-state space of some $k$-local FF Hamiltonian map to a
subsemilattice of $\Theta_k$ under $\mathcal{\vec{L}}_k$.
\end{proposition}
\begin{proof}
A subsemilattice $\mathcal{S}_k$ of $\Theta_k$ is closed under the
join operation, so it must have a largest element $\vec{L}_{M,k}$
such that for any element $\vec{L}_k\in\mathcal{S}_k$,
$\vec{L}_{k}\leq\vec{L}_{M,k}$. Then the MPI of $\vec{L}_{M,k}$
under the map $\mathcal{\vec{L}}_k$ must be the ground-state space
of the $k$-local FF Hamiltonian $\vec{L}_{M,k}^{perp}$. And the MPI
of $\vec{L}_{M,k}$ contains the MPI of any element
$\vec{L}_k\in\mathcal{S}_k$. Therefore the MPIs of $\mathcal{S}_k$
is the ground-state space of $\vec{L}_{M,k}^{\perp}$.

To show the reverse is true, given a $k$-local FF Hamiltonian
$\vec{H}_k$, for any subspace $\mathcal{S}$ of the ground-state
space, $\vec{L}_{k}(\cal{S})$ satisfies
$\vec{L}_{k}(\mathcal{S})\leq\vec{L}_{M,k}$, where $\vec{L}_{M,k}$
is the image of the ground-state space of $\vec{H}_k$ under the map
$\mathcal{\vec{L}}_k$.
\end{proof}
We discuss an example.
\begin{example}
\label{eg:subsemi}
We consider a three-qubit system with qubits $A,B,C$, and a subsemilattice of $\Theta_k$ which contains only two elements $\vec{L}_2$ and $\vec{L}'_2$.
\begin{eqnarray}
&&\vec{L}_2=(\text{span}\{\ket{00},\frac{1}{\sqrt{2}}(\ket{01}+\ket{10})\},\nonumber\\
&&\Span\{\ket{00},\frac{1}{\sqrt{2}}(\ket{01}+\ket{10})\},\text{span}\{\ket{00},\frac{1}{\sqrt{2}}(\ket{01}+\ket{10})\}),
\end{eqnarray}
and
\begin{equation}
\vec{L}'_2=(\text{span}\{\ket{00}\},\text{span}\{\ket{00}\},\text{span}\{\ket{00}\}),
\end{equation}
Apparently $\vec{L}_2$ is the largest element in the subsemilattice.
The pre-image of $\vec{L}_2$ is $\mathcal{S}=\text{span}\{\frac{1}{\sqrt{3}}(\ket{001}+\ket{010}+\ket{100}),\ket{000}\}$, which contains
the MPI of $\vec{L}'_2$ which is $\mathcal{S}'=\text{span}\{\ket{000}\}$. And $\mathcal{S}$ is the ground-state space of the
$2$-local FF Hamiltonian $\vec{L}^{\perp}_2$.
\end{example}

In the language of abstract convexity, these subsemilattices are
actually faces of $\Theta_k$. So Proposition~\ref{pro:sub} is an
analog of the relationship between exposed faces of $\D_k$ and
ground-state spaces of $k$-local Hamiltonians. Furthermore,
Proposition~\ref{pro:sub} shows that indeed for $\Theta_k$, all
faces are exposed. And as atoms are proper analogs of extreme
points, they are then those smallest faces of $\Theta_k$.

As discussed in Proposition~\ref{pro:pure}, if the MPI of an element
$\vec{L}_k\in\Theta_k$ is a single pure state $\ket{\psi}$, then
$\vec{L}_k$ is an atom. In this case, $\ket{\psi}$ is then a unique
ground state of the  $k$-local FF Hamiltonian $\vec{L}_k^{\perp}$.
The reverse is also true, that is, the image of a unique ground
state of a $k$-local FF Hamiltonian under the map
$\mathcal{\vec{L}}_k$ is an atom in $\Theta_k$. However, we do lack
a general understanding of the pre-images for atoms in $\Theta_k$,
so we would like to look further into this. We start from a
definition of `minimal' ground-state space for $k$-local FF
Hamiltonians.

\begin{definition}
A ground-state space of a $k$-local FF Hamiltonian is called minimal if it does not contain any proper subspaces which are ground-state spaces of some other $k$-local FF Hamiltonians.
\end{definition}

Based on this definition we have the following proposition for atoms.
\begin{proposition}
\label{pro:min} The MPI of an atom of $\Theta_k$ under the map
$\mathcal{\vec{L}}_k$ is a minimal ground-state space of $k$-local
FF Hamiltonians. The reverse is also true, that is, the image of a
minimal ground-state space of some $k$-local FF Hamiltonian under
the map $\mathcal{\vec{L}}_k$ is an atom of $\Theta_k$.
\end{proposition}
\begin{proof}
If the MPI $\mathcal{S}$ of an atom of $\Theta_k$ under the map
$\mathcal{\vec{L}}_k$ is not a minimal ground space, then there
exists a proper subspace $\mathcal{S}'$ of $\mathcal{S}$ which is a
ground space of some other $k$-local FF Hamiltonian. Then
$\vec{L}_k(\mathcal{S}')<\vec{L}_k(\mathcal{S})$ which contradicts
the assumption that $\vec{L}_k(\mathcal{S})$ is an atom.

To show the reverse, if the image of a minimal ground-state space $\mathcal{S}$ is not
an atom, then there exists some $\vec{L}'_k<\vec{L}_k(\mathcal{S})$, such that the MPI of $\vec{L}'_k$ is a proper subspace of $\mathcal{S}$, which contradicts the assumption that $\mathcal{S}$ is a minimal ground-state space of $k$-local FF Hamiltonians.
\end{proof}
This proposition is very intuitive. The atoms are minimal elements in $\Theta_k$, so they correspond to minimal ground-state spaces of $k$-local FF Hamiltonians.
However, we still do not understand the detailed structure of the minimal ground-state spaces of $k$-local FF Hamiltonians.
Certainly, unique ground states of $k$-local FF Hamiltonians is a special case of minimal ground-state spaces. For the simplest case for $n$-qubit systems with $k=2$, we show that this is the indeed the only possible case, which is presented as the following theorem.

\begin{theorem}
For an $n$-qubit system, and $k=2$, the MPI of any atom
of $\Theta_2$ is the unique ground state of some $2$-local FF Hamiltonian.
\end{theorem}
\begin{proof}
First of all, for an $n$-qubit system, the ground-state space structure of $2$-local FF Hamiltonian
can be completely characterized, as is shown in~\cite{CCD10,JWZ10}. One important structure theorem
shows that any $2$-local FF Hamiltonian $H$ has a ground state which is a product of
single- or two-qubit states. We then denote this product state by $|\psi\rangle$.
Apparently $|\psi\rangle$ is a unique ground state of some $2$-local FF Hamiltonian.
Then according to Proposition~\ref{pro:pure}, $\vec{L}_2(\ket{\psi})$ is an atom in $\Theta_2$.

Now we further show that any $n$-qubit state $\ket{\psi'}$ which is
not a product of single- or two-qubit states has an image which is
not an atom in $\Theta_2$, under the map $\mathcal{\vec{L}}_2$.
Because $\ket{\psi'}$ is not a product of single- or two-qubit
states, then $\ket{\psi'}$ as a ground state of any $2$-local FF
Hamiltonian must be degenerate with another state $\ket{\psi''}$
which is a product of single- or two-qubit state. Therefore, one
must have $\vec{L}_2(\ket{\psi''})<\vec{L}_2(\ket{\psi'})$. Hence
$\vec{L}_2(\ket{\psi'})$ is not an atom in $\Theta_2$.
\end{proof}

However, this theorem is in general not true. For instance, Kitaev's
toric code on a torus has a four-fold degenerate ground-state
space~\cite{KSV02}, while the corresponding Hamiltonian is $4$-local
and FF. The image of the ground-state space under the map
$\mathcal{\vec{L}}_4$ is an atom in $\Theta_4$, where any subspace
of the ground-state space has exactly the same image under the map
$\mathcal{\vec{L}}_4$. It then remains a challenge to understand
further the possible pre-image structures for atoms in $\Theta_k$.

Finally, we discuss the MPI for join irreducibles of $\Theta_k$.
Intuitively, this MPI is in a sense also `irreducible'. We make this
more concrete by introducing the following definition.
\begin{definition}
The ground-state space of a $k$-local FF Hamiltonian is called irreducible if it cannot be written as a sum of two proper subspaces that are ground-state spaces of some other $k$-local FF Hamiltonians.
\end{definition}
Then the MPI of join irreducibles of $\Theta_k$ can be linked to irreducible ground spaces of $k$-local FF Hamiltonians.
\begin{proposition}
\label{pro:irrg} The MPI of a join irreducible of $\Theta_k$ under
the map $\mathcal{\vec{L}}_k$ is an irreducible ground-state space
of $k$-local FF Hamiltonians.
\end{proposition}
\begin{proof}
If the MPI $\mathcal{S}$ of a join irreducible of $\Theta_k$ under
the map $\mathcal{\vec{L}}_k$ is not an irreducible ground-state
space, then there exist two proper subspaces $\mathcal{S}'$ and
$\mathcal{S}''$ of $\mathcal{S}$ which are ground-state spaces of
some other $k$-local FF Hamiltonians, and
$\mathcal{S}'+\mathcal{S}''=\mathcal{S}$. Then
$\vec{L}_k(\mathcal{S}')<\vec{L}_k(\mathcal{S})$ and
$\vec{L}_k(\mathcal{S}'')<\vec{L}_k(\mathcal{S})$. Therefore
$\vec{L}_k(\mathcal{S}')\vee\vec{L}_k(\mathcal{S}'')\leq\vec{L}_k(\mathcal{S})$.
However, one must have
$\vec{L}_k(\mathcal{S}')\vee\vec{L}_k(\mathcal{S}'')=\vec{L}_k(\mathcal{S})$,
as the MPI of $\vec{L}_k(\mathcal{S}')\vee\vec{L}_k(\mathcal{S}'')$
equals $\mathcal{S}$, due to Lemma~\ref{lm:pre} and
Proposition~\ref{pro:sub}. This contradicts the assumption that
$\vec{L}_k(\mathcal{S})$ is a join irreducible.
\end{proof}
As an example, $\vec{L}_2$ in Example~\ref{eg:subsemi} is a join
irreducible in $\Theta_k$, and the MPI $\mathcal{S}$ is an
irreducible ground-state space of $2$-local FF Hamiltonians. It
remains open whether the reverse is also true. That is, whether the
image of an irreducible ground-state space of some $k$-local FF
Hamiltonian under the map $\mathcal{\vec{L}}_k$ is a join
irreducible of $\Theta_k$.

\section{Conclusion and Discussion}
\label{sec:con}

We have studied the ground-state space properties of
frustration-free Hamiltonians, from a new angle of reduced spaces.
For an $n$-particle system, we discuss the mathematical structure of
$\Theta_k$, the set of all the $k$-RSs. We have provided three
different but closely related perspectives. The first one is the
most straightforward, which is based on a binary operation called
sum, under which $\Theta_k$ is closed and forms an idempotent,
commutative semigroup without zero element. This reveals the most
basic structure of $\Theta_k$. The sum operation is just the usual
sum of vector spaces. However, because the lack of zero element,
which cannot be consistently defined on $\Theta_k$, $\Theta_k$ is
not closed under the usual intersection of vector spaces.

The second one is based on the language of order and semilattice,
which is the most natural as we are indeed studying some special
kind of subspaces of the Hilbert space, given that the set of all
the subspaces of the Hilbert space is a lattice, whose structure has
been widely studied in the field of quantum logic~\cite{Kal83}. It
turns out the idempotent, commutative semigroup is nothing but a
join-semilattice, where the binary operation join is just the sum
operation used in the semigroup characterization. Also, due to the
lack of zero element, the usual intersection operation of two vector
spaces, which in order-theoretic terms is called meet, cannot be
defined on $\Theta_k$. Therefore, unlike the set of all the
subspaces of the Hilbert space, which is a lattice with both join
and meet operations, $\Theta_k$ is not a lattice but only a
join-semilattice. This characterization of $\Theta_k$ from the
order-theoretic point of view allows us to investigate the structure
of $\Theta_k$ further by borrowing general theory of semilattices.
In particular, attentions are given to three most important kind of
elements in a join-semilattice: atoms, join primes and join
irreducibles. Join irreducibles are elements which cannot be written
as a join of other two elements. They are building blocks of
$\Theta_k$, as from the general theory of semilattices we know that
for $\Theta_k$, where the $n$-particle Hilbert space is finite,
every element can be expressed as a finite join of join
irreducibles. Atoms are the smallest elements in $\Theta_k$, which
is also the smallest join irreducibles. And we show there exists no
join prime in $\Theta_k$. As mentioned in Sec.\ref{sec:semilattice},
the characterizations of these special elements in $\Theta_k$ need
more exploration.

For the third one, we show that the semilattice can be interpreted
as an abstract convex structure. This characterization of $\Theta_k$
by a convex structure provides an analog of the characterization of
the set $\mathcal{D}_k$ of all the $k$-RDMs, which is a convex set.
The smallest nonzero elements in $\Theta_k$, namely atoms, are
analogs of extreme points in $\mathcal{D}_k$. However, contrary to
the points in $\mathcal{D}_k$ which can always be weighted sums of
extreme points, we show that the elements $\Theta_k$ may not be able
to be written as join of atoms, indicating a richer structure for
$\Theta_k$. Indeed, one needs to study also the join irreducibles of
$\Theta_k$, as all the elements in $\Theta_k$ can be written as a
finite join of join irreducibles.

Finally, we relate the structure of $\Theta_k$ to ground spaces of
$k$-local FF Hamiltonians: the subsemilattices of $\Theta_k$
correspond to ground spaces of $k$-local FF Hamiltonians; atoms of
$\Theta_k$ correspond to minimal ground-state spaces of $k$-local FF
Hamiltonians; and join irreducibles of $\Theta_k$ correspond to
irreducible ground-state spaces of $k$-local FF Hamiltonians. For
the case of an $n$-qubit system and $k=2$, we show that the MPI of
any atom of $\Theta_2$ corresponds to unique ground state of some
$2$-local FF Hamiltonian. However, the detailed structures of
minimal ground-state spaces and irreducible ground-state spaces
remain open.

Our study of $\Theta_k$ deepens the understanding of ground-state space properties for frustration-free Hamiltonians, from a new angle of reduced spaces. We believe this angle will open up new methods and directions in the study of ground-state space properties for frustration-free Hamiltonians.

\section{Acknowlegement}

We thank Prof. Mingsheng Ying for helpful discussions regarding lattice theory.
JC is supported by NSERC. ZJ acknowledges
support from ARO and NSF of China (Grant Nos. 60736011 and 60721061).
DWK is supported by NSERC Discovery Grant 400160,
NSERC Discovery Accelerator Supplement 400233 and Ontario Early Researcher Award 048142.
ZW acknowledges grant from the Centre for Quantum Technologies,
and the WBS grant under contract no. R-710-000-008-271 and
R-710-000-007-271.  BZ is supported by NSERC Discovery Grant 400500 and CIFAR.

\bibliographystyle{plain}
\bibliography{FFJPA}

\end{document}